\documentclass[11pt]{article}
\usepackage{fullpage}
\usepackage{color}

\usepackage{times}
\usepackage{epsfig}
\usepackage{amsmath}
\usepackage{amssymb}
\usepackage{amstext}
\usepackage{amsmath}
\usepackage{xspace}
\usepackage{theorem}

\newtheorem{theorem}{Theorem}
\newtheorem{definition}{Definition}
\newtheorem{claim}{Claim}
\newtheorem{lemma}[theorem]{Lemma}
\newtheorem{fact}[theorem]{Fact}

\newcommand{\qed}{\mbox{\ \ \ }\rule{6pt}{7pt} \bigskip}

\newcommand{\comment}[1]{}
\newenvironment{proof}{\noindent{\em Proof:}}{\hfill\qed}

\newenvironment{proofof}[1]{\noindent{\em Proof of #1:}}{\hfill\qed}

\newcommand{\T}{{\mathcal T}}

\newcommand{\E}{{\mathcal E}}

\newcommand{\OPT}{*}

\newcommand{\calD}{{\mathcal D}}
\newcommand{\calU}{{\mathcal U}}
\newcommand{\Gap}{{\textrm {Gap}}}
\newcommand{\Cover}{{\textrm {Cover}}}

\newcommand{\calS}{{\mathcal X}}

\title{Threshold Rules for Online Sample Selection}

\author{Eric Bach\thanks{Computer Sciences Dept., University of
    Wisconsin - Madison, {\tt bach@cs.wisc.edu}.}
  \and Shuchi Chawla\thanks{Computer Sciences Dept., University of
    Wisconsin - Madison, {\tt shuchi@cs.wisc.edu}.}
  \and Seeun Umboh\thanks{Computer Sciences Dept., University of
    Wisconsin - Madison, {\tt seeun@cs.wisc.edu}.}  }

\date{}

\begin{document}

\thispagestyle{empty}

\maketitle

\begin{abstract}
We consider the following sample selection problem. We observe in an
online fashion a sequence of samples, each endowed by a quality. Our
goal is to either select or reject each sample, so as to maximize the
aggregate quality of the subsample selected so far. There is a natural
trade-off here between the rate of selection and the aggregate quality
of the subsample. We show that for a number of such problems extremely
simple and oblivious ``threshold rules'' for selection achieve optimal
tradeoffs between rate of selection and aggregate quality in a
probabilistic sense. In some cases we show that the same threshold
rule is optimal for a large class of quality distributions and is thus
oblivious in a strong sense.

\end{abstract}

\newpage
\setcounter{page}{1}

\section{Introduction}

Imagine a heterogeneous sequence of samples from an array of sensors,
having different utilities reflecting their accuracy, quality, or
applicability to the task at hand.  We wish to discard all but the
most relevant or useful samples. Further suppose that selection is
performed online --- every time we receive a new sample we must make
an irrevocable decision to keep it or discard it. What rules can we
use for sample selection? There is a tradeoff here: while we want to
retain only the most useful samples, we may not want to be overly
selective and discard a large fraction. So we could either fix a rate
of selection (the number of examples we want to retain as a function
of the number we see) and ask for the best quality subsample, or fix a
desirable level of quality as a function of the size of the subsample
and ask to achieve this with the fewest samples rejected.

An example of online sample selection is the following ``hiring''
process that has been studied previously. Imagine that a company
wishing to grow interviews candidates to observe their qualifications,
work ethic, compatibility with the existing workforce, etc. How should
the company make hiring decisions so as to obtain the higest quality
workforce possible?  As for the sensor problem, there is no single
correct answer here. Rather a good hiring strategy depends on the rate
at which the company plans to grow---again there is a trade-off
between being overly selective and growing fast. Broder et
al.~\cite{Hiring} studied this hiring problem in a simple setting
where each candidate's quality is a one-dimensional random variable
and the company wants to maximize the average or median quality of its
workforce.

In general performing such selection tasks may require complicated
rules that depend on the samples seen so far. Our main contribution is
to show that in a number of settings an extremely simple class of
rules that we call ``threshold rules'' is close to optimal on average
(within constant factors).

Specifically, suppose that each sample is endowed with a ``quality'',
which is a random variable drawn from a known distribution. We are
interested in maximizing the aggregate quality of a set of samples,
which is a numerical function of the individual qualities.
Suppose that we want to select a subset of $n$ samples out of a total
of $T$ seen.  Let $Q^{\OPT}_{T,n}$ denote the maximum aggregate
quality that can be achieves by picking the best $n$ out of the $T$
samples. Our goal is to design an online selection rule that
approximates $Q^{\OPT}_{T,n}$ in expectation over the $T$ samples. We
use two measures of approximation --- the ratio of the expected
quality achieved by the offline optimum to that achieved by the online
selection rule, $E[Q^{\OPT}_{T,n}]/E[Q_{T,n}]$, and the
expectation of the ratio of the qualities of the two rules,
$E[Q^{\OPT}_{T,n}/Q_{T,n}]$.  Here the expectations are taken
over the distribution from which the sample is drawn. The
approximation ratios are always at least $1$ and our goal is to show
that they are bounded from above by a constant independent of $n$. In
this case we say that the corresponding selection rule is optimal.


To put this in context, consider the setting studied by Broder et
al.~\cite{Hiring}. Each sample is associated with a quality in the
range $[0,1]$, and the goal is to maximize the average quality of the
subsample we pick. Broder et al. show (implicitly) that if the quality
is distributed uniformly in $[0,1]$ a natural \emph{select above the
mean} rule is optimal to within constant factors with respect to the
optimal offline algorithm that has the same selection rate as the
rule. The same observation holds also for the
\emph{select above the median} rule. Both of these rules are adaptive
in the sense that the next selection decision depends on the samples
seen so far. In more general settings, adaptive rules of this kind can
require unbounded space to store information about samples seen
previously.  For example, consider the following 2-dimensional skyline
problem: each sample is a point in a unit square; the quality of a
single point $(x,y)$ is the area of its ``shadow'' $[0,x]\times[0,y]$,
and the quality of a set of points is the area of the collective
shadows of all the points; the goal is to pick a subsample with the
largest shadow. In this case, a natural selection rule is to select a
sample if it falls out of the shadow of the previously seen
points. However implementing this rule requires remembering on average
$O(\log n)$ samples out of $n$ samples seen~\cite{ave-maxima}. We
therefore study non-adaptive selection rules.

We focus in particular on so-called ``threshold rules'' for
selection. A threshold rule specifies a criterion or ``threshold''
that a candidate must satisfy to get selected. Most crucially, the
threshold is determined \textit{a priori} given a desired selection
rate; it depends only on the number of samples picked so far and is
otherwise independent of the samples seen or picked. Threshold rules
are extremely simple oblivious rules and can, in particular, be
``hard-wired'' into the selection process. This suggests the following
natural questions. When are threshold rules optimal for online
selection problems? Does the answer depend on the desired rate of
selection? We answer these questions in three different settings in
this paper.

The first setting we study is a single-dimensional-quality setting
similar to Broder et al.'s model. In this setting, we study threshold
rules of the form {\em ``Pick the next sample whose quality exceeds
$f(i)$''} where $i$ is the number of samples picked so far. We show
that for a large class of functions $f$ these rules give constant
factor approximations. Interestingly, our threshold rules are optimal
in an almost distribution-independent way. In particular, every rule
$f$ in the aforementioned class is simultaneously constant-factor
optimal with respect to any ``power law'' distribution, and the
approximation factor is independent of the parameters of the
distribution. In contrast, Broder et al.'s results hold only for the
uniform distribution\footnote{While Broder et al.'s result can be
extended to any arbitrary distribution via a standard tranformation
from one space to another, the resulting selection rule becomes
distribution dependent, e.g., ``select above the mean'' is no longer
``select above the mean'' w.r.t. the other distribution upon applying
the transformation.}.

In the second setting, samples are nodes in a rooted infinite-depth
tree. Each node is said to cover all the nodes on the unique path from
the root to itself. The quality of a collection of nodes is the total
number of distinct nodes that they collectively cover. This is
different from the first setting in that the quality defines only a
partial order over the samples. Once again, we study threshold rules
of the form {\em ``Pick the next sample whose quality exceeds
$f(i)$''} and show that they are constant factor optimal.

Our third setting is a generalization of the skyline problem described
previously. Specifically, consider a domain $X$ with a probability
measure $\mu$ and a partial ordering $\prec$ over it. For an element
$x\in X$, the ``shadow'' or ``downward closure'' of $x$ is the set of
all the points that it dominates in this partial ordering, $\calD(x)
= \{y: y\prec x\}$; likewise the shadow of a subset $S\subseteq X$ is
$\calD(S) = \cup_{x\in S} \calD(x)$. Once again, as in the second
setting, we can define the coverage of a single sample to be the
measure of all the points in its shadow. However, unlike the tree
setting, here it is usually easy to obtain a constant factor
approximation to coverage---the maximum coverage achievable is $1$
(i.e. the measure of the entire universe), whereas in many cases
(e.g. for the uniform distribution over the unit square) a single
random sample can in expectation obtain constant coverage. We
therefore measure the quality of a subsample $S\subset X$ by its
``gap'',
$\Gap(S)= 1 - \mu(\calD(S))$. In this setting, rules that place a
threshold on the quality of the next sample to be selected are not
constant-factor optimal. Instead, we study threshold rules of the form
{\em ``Pick the next sample $x$ for which $\mu(\calU(x))$ is at most
$f(i)$''}, where $\calU(x) = \{y: x\prec y\}$ is the set of all
elements that dominate $x$, or the ``upward closure'' of $x$, and show
that these rules obtain constant factor approximations.

\subsection{Related work}

As mentioned earlier, our work is inspired by and extends the work of
Broder et al.~\cite{Hiring}. Broder et al. consider a special case of
the one-dimensional selection problem described above. They assume
that the quality of a sample is distributed uniformly over the
interval $(0,1)$; this assumption is not without loss of generality.
They analyze two adaptive selection rules---\emph{select above the
  mean}, and \emph{select above the median}---and show that both are
constant-factor optimal , although they lead to different growth
rates. These rules are adaptive in the sense that the next selection
decision depends on the quality of the samples accepted so far.  Note
that the \emph{select above the median} rule requires the algorithm to
remember all of the samples accepted so far, and is therefore a
computationally intensive rule. Even the relatively simpler
\emph{select above the mean} rule requires remembering the current
mean and number so far accepted. In contrast we show
(Section~\ref{sec:one-dim}) that there exists a class of simple
non-adaptive selection strategies that also achieves optimality and
includes rules with selection rates equal to those of the ones studied
by Broder et al. These strategies make decisions based only on the
number hired so far. Furthermore we extend these results to more
general coverage problems.

Our third setting is closely related to the skyline problem that has
been studied extensively in online settings by the database community
(see, for example, \cite{all-skyline-probs} and references
therein). Kung, et al.~\cite{finding-maxima} gave an offline
divide-and-conquer algorithm that finds the skyline of a given set of
vectors in $d$-dimensional space. Their algorithm uses $O(n \log_2 n)$
comparisons of vector components when $d = 2,3$ and $O(n (\log_2
n)^{d-2})$ when $d \geq 4$. The implementation of a \emph{Skyline}
query for database systems was recently introduced by
\cite{skyline-operator}. The closest in spirit to our work is
\cite{prob-skyline-sliding}. They considered a stream of uncertain
objects to model uncertainty in measurement. Each object has an
associated set of possible instances and they are interested in the
objects whose probability of being dominated by another object is at
most some $q$ supplied by the database user. 


Online sample selection is closely related to secretary problems,
however there are some key differences. In secretary problems (see,
e.g., \cite{Fer89,Fre83,Sam91}) there is typically a fixed bound on
the desired number of hires. In our setting the selection process is
ongoing and we must pick more and more samples as time passes. This
makes the tradeoff between the rate of hiring and the rate of
improvement of quality interesting.

Finally, while our goal is to analyze a class of online algorithms in
comparison to the optimal offline algorithms, our approach is
different from the competitive analysis of online algorithms
\cite{BE-Y}. In competitive analysis the goal is to perform nearly as
well as the optimal offline algorithm for {\em any arbitrary} sequence
of input. In contrast, we bound the {\em expected} competitive ratio
of the rules we study. Furthermore, a crucial aspect of the strategies
that we study is that not only are they online, but they are also
non-adaptive or oblivious. That is, the current acceptance threshold
does not depend on the samples seen by the algorithm so far. In this
sense, our model is closer in spirit to work on oblivious algorithms
(see, e.g., \cite{univ-tsp, obl-routing, obl-nd}). Oblivious
algorithms are highly desirable in practical settings because the
rules can be hard-wired into the selection process, making them 
very easy
to implement. The caveat is, of course, that for many
optimization problems oblivious algorithms do not provide good
approximations. Surprisingly, we show that in many scenarios related
to sample selection, obliviousness has only a small cost.

\section{Models and results}
\label{sec:model}
Let $X$ be a domain with probability measure $\mu$ over it. A {\em
threshold rule} $\calS$ is specified by a sequence of subsets of $X$
indexed by 
$\mathbb{N}$:
$X=\calS_0\supseteq \calS_1\supseteq \calS_2 \supseteq \cdots
\supseteq \calS_n \supseteq \cdots$. A sample is selected if it belongs to
$\calS_i$ where $i$ is the number of samples previously selected.

Let $\T$ be an infinite sequence of samples drawn i.i.d. according to
$\mu$. Let $\T^\calS(n)$ denote the prefix of $\T$ such that the last
sample on this prefix is the $n$th sample chosen by the threshold rule
$\calS$; let $T^\calS_n$ denote the length of this prefix.
We drop the superscript and the subscript when they are clear from the
context. The ``selection overhead'' of a threshold rule as a function
of $n$ is the expected waiting time to select $n$ samples, or
$E[T^\calS_n]$, where the expectation is over $\T$.

Let $Q$ be a function denoting ``quality''. Thus $Q(x)$ denotes the
quality of a sample $x$ and $Q(S)$ the aggregate quality of a set
$S\subset X$ of samples. $Q(x)$ is a random variable and we assume
that it is drawn from a known distribution. Let $Q^*_{\T,n}$ denote
the quality of an optimal subset of $n$ out of a set $\T$ of samples
with respect to measure $Q$. We use $Q^*_{n}$ as shorthand for
$Q^*_{\T^\calS(n),n}$ where $\calS$ is clear from the context. Let
$Q^\calS_{\T^\calS(n),n}$ ($Q_{n}$ for short) denote the quality of a
sample of size $n$ selected by threshold rule $\calS$ with respect to
measure $Q$.

We look at both maximization and minimization problems. For
maximization problems we say that a threshold rule $\calS$ achieves a
{\em competitive ratio of $\alpha$ in expectation} with respect to $Q$
if for all $n$,
\[ E_{\T \sim \mu}\left[ \frac{Q^*_{n}}{Q_{n}} \right] \le \alpha\]

\noindent
Likewise, $\calS$ {\em $\alpha$-approximates expected quality} with
  respect to $Q$ if for all $n$,
\[\frac{E_{\T \sim \mu}[Q^*_{n}]}{E_{\T \sim \mu}[Q_{n}]} \le \alpha\]

\noindent
For minimization problems, the ratios are defined similarly:
\[ \text{Exp. comp. ratio} = \max_n E_{\T \sim \mu}\left[ \frac{Q_{n}}{Q^*_{n}} \right]; 
\hfill
\text{Approx. to exp. quality} = \max_n \frac{E_{\T \sim \mu}[Q_{n}]}{E_{\T \sim \mu}[Q^*_{n}]} \]

\noindent
We now describe the specific settings we study and the results we obtain.

\paragraph{Model 1: Unit interval (Section~\ref{sec:one-dim}).}
Our first setting is the one-dimensional setting studied by Broder et
al.~\cite{Hiring}. Specifically, each sample is associated with a
quality drawn from a distribution over the unit line. Our measure of
success is the mean quality of the subsample we select. Note that in
the context of approximately optimal selection rules this is a weak
notion of success. For example when $\mu$ is the uniform distribution
over $[0,1]$, even in the absence of any selection rule we can achieve
a mean quality of $1/2$, while the maximum achievable is $1$. So
instead of approximately maximizing the mean quality, we approximately
minimize the {\em mean quality gap}---$\Gap(S) = 1-(\sum_{x\in
S}x)/|S|$---of the subsample.

We focus on {\em power-law} distributions on the unit line,
i.e. distributions with c.d.f. $\mu(1-x)=1-x^k$ for some constant $k$,
and study threshold rules of the form $\calS_i=\{x : x\ge 1-c_i\}$
where $c_i = \Omega(1/\textrm{poly}(i))$. We show that these threshold
rules are constant factor optimal simultaneously for {\em any}
power-law distribution. Remarkably, this gives an optimal selection
algorithm that is oblivious of even the underlying
distribution. Formally we obtain the following result.

\begin{theorem}
\label{thm:one-dimensional}
For the unit line equipped with a power-law distribution, any
threshold rule $\calS_i=\{x : x\ge 1-c_i\}$, where $c_i =
1/i^\alpha$ with $0 \leq \alpha < 1$ for all $i$, achieves an $O(1)$
approximation to the expected gap, where the constant in the $O(1)$
depends only on $\alpha$ and not on the parameters of the
distribution.
\end{theorem}

\paragraph{\bf Dominance and shadow.} For the next two settings, we
need some additional definitions. Let $\prec$ be a partial order over
the universe $X$. As defined earlier, the shadow of an element $x\in
X$ is the set of all the points that it dominates, $\calD(x) = \{y:
y\prec x\}$; likewise the shadow of a sample $S\subseteq X$ is
$\calD(S) = \cup_{x\in S} \calD(x)$. Let $\calU(x) = \{y :
x \prec y\}$ be the set of points that shadow $x$; $\calU(S)$ for a
set $S$ is defined similarly. Note that $\calU(x)$ is a subset of
$X\setminus\calD(x)$ and $\mu(\calU(x))$ is the probability that
a random sample covers $x$.

\paragraph{\bf Model 2: Random tree setting (Section~\ref{sec:tree}).}
While the previous setting was in a continuous domain, next we
consider a discrete setting, where the goal is to maximize the
cardinality of the shadow set. Specifically, our universe $X$ is the
set of all nodes in an rooted infinite-depth binary tree. The
following random process generates samples. Let $0 < p < 1$. We start
at the root and move left or right at every step with equal
probability. At every step, with probability $p$ we terminate the
process and output the current node. A node $x$ in the tree dominates
another node $y$ if and only if $y$ lies on the unique path from the
root to $x$. For a set $S$ of nodes, we define coverage as $\Cover(S)
= |\calD(S)|$. Note, that unlike in the previous setting, there is
no notion of a gap in this setting.

Once again the threshold rules we consider here are based on sequences
of integers $\{c_i\}$. For any such sequence, we define $\calS_i = \{
x : |\calD(x)|\ge c_i\}$. We show that constant-factor optimality
can be achieved with exponential or smaller selection overheads.




\begin{theorem}
\label{thm:tree}
For the binary tree model described above, any threshold rule based on
a sequence $\{c_i\}$ with $c_i=O({\mbox {poly}}(i))$ achieves an
$O(1)$ competitive ratio in expectation with respect to coverage, as
well as an $O(1)$ approximation to the expected coverage.
\end{theorem}

\paragraph{\bf Model 3: Skyline problem (Section~\ref{sec:two-dim}).}
Finally, we consider another continuous domain that is a
generalization of the skylike problem mentioned previously. We are
interested in selecting a set of samples with a large
shadow. Specifically, we define 
the ``gap'' of $S$ to be $\Gap(S) = 1-\mu(\calD(S))$. Our goal is to
minimize the gap.

We show that a natural class of threshold rules obtains near-optimal
gaps in this setting. Recall that $\calU(x)$ for an element $x\in
X$ denotes the set of elements that dominate $x$. We consider
threshold rules of the form $\calS_i=\{x\in
X : \mu(\calU(x))\le c_i\}$ for some sequence of numbers
$\{c_i\}$. We require the following continuity assumption on the
measure $\mu$.
\begin{definition}
{\bf (Measure continuity)} For all $x\in X$ and $c\in
[0,\mu(\calU(x))]$, there exists an element $y\in\calU(x)$ such
that $\mu(\calU(y))=c$. Furthermore, there exist elements
$\underline{x},\overline{x}\in X$ with $\calU(\underline{x})=X$ and
$\calU(\overline{x})=\emptyset$.
\end{definition}

\noindent
Measure continuity ensures that the sets $\calS_i$ are all non-empty
and proper subsets of each other.




\begin{theorem}
\label{thm:skyline-arbitrary}
For the skyline setting with an arbitrary measure satisfying measure
continuity, any threshold rule based on a sequence $\{c_i\}$ with
$c_i=i^{-(1/2-\Omega(1))}$
achieves a $1+o(1)$ competitive ratio in
expectation with respect to the gap.
\end{theorem}

We note that the class of functions $c_i$ specified in the above
theorem includes all functions for which $1/c_i$ grows
subpolynomially. In particular, this includes threshold rules with
selection overheads that are slightly superlinear.

For the special case of the skyline setting over a two-dimensional
unit square $[0,1]^2$ bestowed with a product distribution and the
usual precedence ordering---$(x_1,y_1)\prec (x_2,y_2)$ if and only if
$x_1\le x_2$ and $y_1\le y_2$---we are able to obtain a stronger
result that guarantees constant-factor optimality for any polynomial
selection overhead:
\begin{theorem}
\label{thm:skyline-product}
For the skyline setting on the unit square with any product
distribution, any threshold rule based on a sequence $\{c_i\}$ with
$c_i=\Omega(1/{\mbox {poly}}(i))$ achieves a $1+o(1)$ competitive ratio in
expectation with respect to the gap.
\end{theorem}



\section{Sample selection in one dimension}
\label{sec:one-dim}
\label{opt-exp-mean-gap:sec}
\newcommand{\barQ}{\overline{Q}}
\newcommand{\barG}{\overline{G}}
\newcommand{\barD}{\overline{D}}
\newcommand{\barH}{\overline{H}}
\newcommand{\barh}{\overline{h}}
\newcommand{\barq}{\overline{q}}
\newcommand{\barg}{\overline{g}}
\newcommand{\bard}{\overline{d}}
\newcommand{\barmu}{\overline{\mu}}
\newcommand{\barX}{\overline{X}}
\newcommand{\barY}{\overline{Y}}
\newcommand{\barx}{\overline{x}}




We will now prove Theorem~\ref{thm:one-dimensional}.
For a (random) variable $x \in [0,1]$, let $\barx$ denote its
complement $1-x$. For a cumulative distribution $\mu$ with domain
$[0,1]$, we use $\barmu$ to denote the cumulative distribution for
the complementary random variable: $\barmu(x)=1-\mu(1-x)$.


Let $Y$ denote a draw from the power-law distribution $\mu$, and $Y_n$
denote the (random) quality of the $n$th sample selected by
$\calS$. Note that since $\mu$ is a power-law distribution, $\barY_i$
is statistically identical to $c_i\barY$. Then, the mean quality gap
of the first $n$ selected samples is given by
$\Gap_n= \frac{1}{n}\sum_{i=1}^n \barY_i = \frac{1}{n}\sum_{i=1}^n
c_i\barY$, and, by linearity of expectation we have 
\begin{equation}
  \label{eq:gap}
  E[\Gap_n] = \frac{E[\barY]}{n}\sum_{i=0}^{n-1} c_i\enspace.
\end{equation}

On the other hand, the following lemma gives the optimal mean quality
achievable when we pick a subsample of size $n$ from a sample of size
$E[T_n]$. See Appendix~\ref{sec:appendix} for the proof.
\begin{lemma}
\label{fact:opt-mean-gap}
The expected mean gap of the largest $n$ out of $E[T_n]$ samples drawn from
a distribution with $\barmu(x) = x^k$ is 
  \begin{equation}
    \label{eq:largest}
    \frac{1}{1 + 1/k} \left( \frac{n}{E[T_n] + 1} \right)^{1/k}\enspace .
  \end{equation}
\end{lemma}

First, we bound (\ref{eq:gap}) in terms of (\ref{eq:largest}) by noting
that the expected selection overhead of $\calS$ is given by 
\[E[T_n] =
\sum_{i=1}^n 1/\barmu(c_i) = \sum_{i=1}^n 1/c_i^k\enspace.\]

\begin{lemma}
\label{thm:one-dim-euler}
  For selection thresholds $c_i = 1/i^\alpha$ with $0 \leq \alpha < 1$, we
  have
  \[E[\Gap_n] \leq \frac{1}{1-\alpha} 
  \left( \frac{n}{E[T_n]} \right)^{1/k}\enspace .\]
\end{lemma}

\begin{proof}
  The proof follows from the Euler-Maclaurin formula and the fact that
  $E[\barX] \leq 1$.
  \begin{align*}
    E[\Gap_n]\cdot (E[T_n])^{1/k}
    &= \left(\frac{E[\barX]}{n}\sum_{i=1}^n i^{-\alpha} \right) 
    \left( \sum_{i=1}^n i^{\alpha k} \right)^{1/k} \\
    &\approx \frac{E[\barX]}{n}\cdot\frac{n^{1-\alpha}}{1-\alpha}
    \cdot \left(\frac{n^{ak+1}}{\alpha k + 1}\right)^{1/k}\\
    &\leq 
    \frac{n^{1/k}}{1-\alpha}\enspace .
  \end{align*}
\end{proof}

Lemmas~\ref{fact:opt-mean-gap} and \ref{thm:one-dim-euler} together
show that the expected mean gap of these thresholds rules is only a
small constant factor bigger than the optimal offline selection rule
that picks the best $n$ out of $E[T_n]$ samples. Finally we show that
these threshold rules are in fact constant factor optimal in the
following stricter sense: if an adversary were allowed to choose any
$n$ out of $T_{n+1} - 1$, its expected mean gap is only a constant
factor smaller than that of the online algorithm. We denote this
optimal offline gap by $\Gap^\OPT_{n+1}$.
We are interested
in $T_{n+1} - 1$ is because the adversary should be able to use the
samples we rejected while we were waiting for the $(n+1)$th selection.

\begin{lemma}
For $c_i$ satisfying $c_i = 1/i^\alpha$ with $0 \leq \alpha < 1$ for all $i$, we have
  \[E[\Gap^\OPT_{n+1}] \geq \frac{1}{16} 
  \left(\frac{n}{E[T_n]}\right)^{1/k}\enspace .\]
\end{lemma}

\begin{proof}
  By Markov's inequality we have
  \begin{align*}
    E[\Gap^*_n] 
    &\geq \frac{1}{2} E[\Gap^*_n : T_{n+1} 
    \leq 2 E[T_{n+1}]]\\
    &\geq \frac{1}{2}\frac{1}{(1 + 1/k)}\cdot 
    \left(\frac{n}{2E[T_{n+1}]}\right)^{1/k} \\
    &\approx \left(\frac{1}{2}\right)^{1+1/k}\left(\frac{1}{1+1/k}\right)
    \left(\frac{n}{E[T_n]}\right)^{1/k}\left(1+\frac 1n\right)^{-\alpha}\\
    &\geq \frac{1}{16}\left(\frac{n}{E[T_n]}\right)^{1/k} \enspace .
\end{align*}
\end{proof}

 \noindent
 Together with Lemma~\ref{thm:one-dim-euler}, this proves Theorem
 \ref{thm:one-dimensional} with the constant in the $O(1)$ equal to
 $16/(1 - \alpha)$.

\section{Sample selection in binary trees}
\label{sec:tree}
We prove Theorem~\ref{thm:tree} in two parts: (1) the ``fast-growing
thresholds'' case, that is, $c_i=O({\mbox {poly}}(i))$ and $c_i\ge\log
i$ for all $i$, and, (2) the ``slow-growing thresholds'' case, that
is, $c_i \leq c_{i/2} + O(1)$ for all $i$. 

We begin with some notation and observations. Recall that for a node
$x$ in the tree $\calD(x)$ denotes both the unique path from the root
to $x$ as well as the set of nodes covered by $x$ (the shadow of
$x$). Let $\calD_k(x)$ be the $k$th node on $\calD(x)$, $\calD_{\leq
k}$ be the first $k$ nodes of $\calD(x)$, and $\calD_{\geq k}
= \calD(x) \setminus \calD_{< k}$.



We say that a set of $n$ paths associated with nodes $x_1, \ldots,
x_n$ is \emph{independent} at level $k$ if $|\cup_{i=1}^n
\calD_k(x_i)| = n$. That is, no two paths share the same vertex at
level $k$, and are disjoint after level $k$.  We have the following
fact


\begin{fact}
\label{fact:disjoint-paths}
If a set of $n$ paths $\{\calD(x_i)\}$, of length $\ge k'$ each, is
independent at level $k < k'$, then
  \[|\calD(\{x_1,\ldots,x_n\})| = 
  \left|\bigcup_{i = 1}^n \calD(x_i)\right| \geq
  \left|\bigcup_{i = 1}^n \calD_{\geq k}(x_i)\right|\geq n(k' - k)\enspace .\]
\end{fact}


Our analysis depends on whether $c_i$ is a slow-growing or
fast-growing function. We first consider the case of
$c_i=O({\mbox{poly }} n)$ but with $c_i\ge \log_2 i$ for all $i$.

\begin{theorem}
\label{thm:tree-fast}
For the binary tree model described above, any threshold rule based on
a sequence $\{c_i\}$, with $c_i=O({\mbox {poly}}(i))$ and $c_i\ge\log
i$ for all $i$, achieves an $O(1)$ competitive ratio in expectation
with respect to coverage, as well as an $O(1)$ approximation to the
expected coverage.
\end{theorem}


\begin{proof}
Let $f(i)=c_i-\log i$. We will first obtain an upper bound on
  $\Cover^\OPT_n$. Let $S_n$ be the $n$ selected nodes, $O_n$ the
  optimal set of $n$ paths, and $R_n$ the paths that are rejected and
  are not covered by $S_n$.
  \begin{align*}
    \Cover^\OPT_n &= |(\calD(O_n) \cap \calD(R_n)) \cup (\calD(O_n) \cap \calD(S_n))|\\
    &\leq |\calD(O_n) \cap \calD(R_n)| + |\calD(S_n)|\\
    &\leq (2n + nf(n)) + \Cover_n\enspace,
  \end{align*}
  Here the last inequality follows by noting that
  $\calD(O_n) \cap \calD(R_n)$ forms a binary tree with at most
  $2n$ vertices in the first $\log n$ levels, and at most $nf(n)$
  other vertices since it is the union of $n$ paths of length at most
  $\log n + f(n)$.

 Next we obtain a lower bound on $\Cover_n$. Consider the last
 $n/2$ selected nodes $s_{n/2 + 1}, \ldots, s_n$ and their paths 
 $D(s_{n/2+1}), \ldots, D(s_n)$. By definition,
  $|\calD(s_{n/2 + i})| \geq c_{n/2} = \log n/2 + f(n/2)$. Let $N =
  |\cup_{i=1}^{n/2}\calD_{\log n/2}(s_{n/2+i})|$ be the number of paths  $\calD(s_{n/2 + i})$
  that are independent at
  level $\log n/2$. 
  Since $\calD_{\log n/2}(s_{n/2 + i})$ chooses from each of the
  $n/2$ nodes at level $\log n/2$ equiprobably, $N$ has the same
  distribution as the number of occupied bins when $n/2$ balls are
  thrown into $n/2$ bins uniformly at random. The expected number of
  unoccupied bins is $n/2e$. By Markov's inequality, with probability
  at least $1/2$, the number of empty bins is at most $2\frac{n}{2e}$.
  So, we have that $\Pr[N \geq \frac{n}{2}(1 - 2/e)] \geq 1/2$. Thus, we have
  \begin{align*}
    E\left[\frac{\Cover^\OPT_n}{\Cover_n}\right]
    &\leq E\left[\frac{\Cover^\OPT_n}{\Cover_n} :  
      N \geq \frac{n}{2}(1 - 2/e) \right] + 
          \Pr[N \leq \frac{n}{2}(1 - 2/e)]\\
    &\leq E\left[\frac{n(2 + f(n)) + \Cover_n}{\Cover_n} :  
      N \geq \frac{n}{2}(1 - 2/e) \right] + \frac{1}{2}\\
    \intertext{and by Fact \ref{fact:disjoint-paths},}
    &\leq \frac{3}{2} + \frac{2(2 + f(n))}{(1 - 2/e)f(n/2)}
    = O(1)\enspace ,
  \end{align*}
  where the constant depends on $f(n)$.
  Likewise we can obtain a bound on the approximation factor by noting that  
  \begin{align*}
    E\left[\Cover_n\right]
    &\geq E\left[\Cover_n  :  
      N \geq \frac{n}{2}(1 - 2/e) \right]\cdot
          \Pr[N \geq \frac{n}{2}(1 - 2/e)]\\
    &\geq \frac{1}{2}\frac{n}{2}(1 - 2/e)f(n/2)\\
  \end{align*}
  and therefore,
  \begin{align*}
    \frac{E\left[\Cover^\OPT_n\right]}{E\left[\Cover_n\right]}
    &\le 1+ \frac{2n+nf(n)}{\frac{1}{2}\frac{n}{2}(1 - 2/e)f(n/2)} = O(1)
  \end{align*}
  where once again the constant depends on $f(n)$.
\end{proof}

\noindent
Next we consider the case when $c_i$ is a slow-growing function. In
particular we assume that $c_i \leq c_{i/2}+O(1)$ for all $i$.

\begin{theorem}
  For the binary tree model described above, any threshold rule based
  on a sequence $\{c_i\}$, with $c_i \leq c_{i/2} + O(1)$ for all $i$,
  achieves an $O(1)$ competitive ratio in expectation with respect to
  coverage, as well as an $O(1)$ approximation to the expected
  coverage.
\end{theorem}


\begin{proof}
  We follow the outline of the previous proof. Consider the last $n/2$
  selected nodes $s_{n/2 + 1}, \ldots, s_n$. If $c_{n/2}\ge\log(n/2)$,
  we consider the number of independent paths at level $\log(n/2)$ and
  the proof goes through exactly as before. So suppose that
  $c_{n/2}<\log (n/2)$. Let $N =
  |\cup_{i=1}^{n/2} \calD_{c_{n/2}}(s_{n/2 + i}) |$. Then there are
  at most $2^{c_{n/2}} < n/2$ nodes at level $c_{n/2}$, so we can use
  the same balls-and-bins argument as in the previous proof to obtain
  $\Pr[N \geq 2^{c_{n/2}}(1 - 2/e)] \geq 1/2$.
 
  Since there are at most $2^{c_n + 1}$ nodes in the first $c_n$
  levels of a binary tree, and $\Cover_n \geq N$, we have that
  \begin{align*}
    E\left[\frac{\Cover^\OPT_n}{\Cover_n}\right]
    &\leq E\left[\frac{2^{c_n+1} + \Cover_n}{\Cover_n}  :  
      N \geq 2^{c_{n/2}}(1 - 2/e) \right] + \frac{1}{2}\\
    &\leq \frac{3}{2} + \frac{2^{c_n + 1}}{2^{c_{n/2}}(1 - 2/e)} 
    = O(1)
  \end{align*}
  We can also use the same argument as in the previous proof to prove
  the claimed bound on the approximation factor:
  \[    
  \frac{E\left[\Cover^\OPT_n\right]}{E\left[\Cover_n\right]}
  \leq 1 + \frac{2^{c_n + 1}}{\frac{1}{2}\cdot 2^{c_{n/2}}\cdot(1 - 2/e)}
  = O(1)\enspace.
  \]
\end{proof}

\section{Sample selection in the skyline model}
\label{sec:two-dim}
In this section we focus on the following ``skyline'' model. We first
consider the case where the universe $X$ is the unit square $[0,1]^2$
and $(x_1,y_1)\prec (x_2,y_2)$ for $(x_1,y_1), (x_2,y_2)\in X$ if and
only if $x_1\le x_2$ and $y_1\le y_2$. In
Section~\ref{sec:skyline-arbitrary} we discuss general spaces.

\subsection{Uniform and product distributions over $2$ dimensions}
As mentioned earlier, we consider threshold rules of the form
$\calS_i=\{(x,y)\in X : 
\mu(\calU(x,y))\le c_i\}$ for some sequence of numbers
$\{c_i\}$, where $\calU(x,y)$ is the set of points that dominate
$(x,y)$.

For simplicity, we first prove the following version of
Theorem~\ref{thm:skyline-product} for the uniform distribution over
$X$, and then describe how it extends to general product
distributions.

\begin{theorem}
\label{thm:skyline-uniform}
For the skyline setting on the unit square with the uniform
distribution, any threshold rule based on a sequence $\{c_i\}$ with
$c_i=\Omega(1/{\mbox {poly}}(i))$ achieves a $1+o(1)$ competitive ratio in
expectation with respect to the gap.
\end{theorem}


Let $S_n$ denote the set of samples selected by the (implicit)
threshold rule out of the set $T_n$ of samples seen. Let $R_n =
T_n\setminus S_n$ denote the samples rejected by the threshold rule,
and $O_n$ denote an optimal subset of $T_n$ of size $n$. Recall that
our goal is to maximize the shadow of the selected subsample, and so
all points in $O_n$ must be undominated by other points. Let $\E_n$
denote the event that $O_n$ contains a point in $R_n$, that is, there
is a point in $R_n\setminus\calD(S_n)$. It is immediate that
$\Gap_n\ne \Gap^*_n$ if and only if the event $\E_n$ happens. We will
show that the event $\E_n$ occurs with very low probability and use
this fact to prove Theorem~\ref{thm:skyline-uniform}.


We first show how the approximation factor and expected competitive
ratio with respect to the gap of a threshold rule relates to the
probability of the event $\E_n$.
 
\begin{lemma}
\label{lem:skyline-OPT}
For the skyline model with an arbitrary distribution, 
the gap of a
threshold rule based on the sequence $\{c_i\}$ satisfies the
following, where $\E_n$ is the event that $\Gap_n\ne \Gap^*_n$, we have
\[
 E_\mu\left[\frac{\Gap_n}{\Gap^*_n}\right] 
 \leq 1 + \frac{1}{c_n}\Pr[\E_n]\enspace. \]  
\end{lemma}
\begin{proof}
We apply Bayes' rule to get
\[E_\mu\left[\frac{\Gap_n}{\Gap^*_n}\right] 
 \leq 1 + E_\mu\left[\frac{\Gap_n}{\Gap^*_n} : \E_n\right]\Pr[\E_n]\enspace. \]
  If the event $\E_n$ happens then by definition $O_n$ contains a
  point in $R_n$, say $x$. Then, $\Gap^*_n
  = \Gap(O_n) \geq \mu(\calU(x)) > c_n$, where the last
  inequality follows from noting that $x\not\in \calS_n$. On the other
  hand, $\Gap_n$ is always less than $1$. Therefore the claim follows.
\end{proof}


To complete the proof of Theorem~\ref{thm:skyline-uniform} we give an
upper bound on the probability of the event $\E_n$. Our goal is to
show that with high probability, every sample in $R_n$ is dominated by
some sample in $S_n$. We start with a simple observation about the
number of rejected samples.

\begin{fact}
\label{fact:rejects-bound}
$E[|R_n|] \leq E[T_n] = \sum_{i=1}^n 1/\mu(\calS_i) \le n/c_n$.
\end{fact}

\begin{fact}
Let $\mu$ be the uniform measure over $[0,1]^2$. Then for all $n$,
$\mu(\calS_n)= c_n(1+\ln 1/c_n)\enspace$.
\end{fact}

\noindent
We first note that many of the samples in $S_n$ are in fact in
$\calS_n$.
\begin{lemma}
\label{lem:const-frac}
Let $\alpha$ be a constant satisfying $\frac{c_i}{c_{i/2}}\ge \alpha$ for
large enough $i$. Then with probability $1-o(c_n)$, $\alpha n/4$ of
the samples in $S_n$ belong to $\calS_n$.
\end{lemma}
\begin{proof}
Consider samples belonging to $S_n\cap\calS_{n/2}$; these are at least
$n/2$ in number. We claim that a constant fraction of these are in
$\calS_n$ with high probability. In particular, \[\Pr_{x \sim \mu}
[x \in \calS_{n}  :  x \in \calS_{n/2 + i}]
= \frac{\mu(\calS_{n})}{\mu(\calS_{n/2 +
i})} \geq \frac{\mu(\calS_{n})}{\mu(\calS_{n/2})} \geq \alpha\enspace, \] Here
we used the fact that $\frac{1+\ln 1/c_n}{1+\ln 1/c_{n/2}} \ge 1$.
Therefore, the expected number of samples in $S_n\cap\calS_n$ is at
least $\alpha n/2$. Since this is a binomial random variable, by
Chernoff bounds,
\[\Pr\left[ |S_n\cap\calS_n| <  \frac{1}{2}\alpha(n/2)\right] < 
\exp\left\{-\alpha(n/2)\left(\frac{1}{2}\right)^2/2\right\} = o(c_n)\enspace,\]
since $c_i=\Omega(1/{\mbox {poly}}(i))$.
\end{proof}

The following lemma shows that given sufficient number of samples in
$S_n$ belonging to $\calS_n$, with high probability $R_n$ is dominated
by these samples. For the next lemma, let $\E'_n$ denote the event
that at least one point in $R_n$ is not dominated by $S_n\cap\calS_n$
and let $z=|S_n\cap\calS_n|$.

We first state the following consequence of measure continuity.

\begin{fact}
\label{fact:conseq-meas-cont}
  Let $\mu$ satisfy measure continuity. Then for all
  $k \in \mathbb{N}$, and for all $y \notin \calS_k$, we have
  $\mu(\calU(y) \cap \calS_k) \geq c_k$.
\end{fact}
\begin{proof}
By measure continuity, there exists a $z \in \calU(y)$ such that
$\mu(\calU(z)) = c_k \leq \mu(\calU(y))$. Thus, we have that $z
\in \calS_k$, $\calU(z)\subseteq \calU(y) \cap \calS_k$, and so
$\mu(\calU(y) \cap \calS_k) \geq \mu(\calU(z)) = c_k$.
\end{proof}

\begin{lemma}
\label{lem:dominate-rejects}
Conditioned on $z$, we have
\[\Pr[\E'_n] \leq \exp\left\{-\frac{z c_n}{\mu(\calS_n)}\right\}\cdot E[|R_n|]\enspace.\]
\end{lemma}
\begin{proof} 
  For any sample $y$ in $R_n$, the probability that it is dominated by
  a uniformly random sample in $\calS_n$ is 
  \begin{align*}
    \Pr_{x \sim \mu} [y \in \calD(x)  :  x \in \calS_n] 
    = \frac{\mu(\calU(y)\cap\calS_n)}{\mu(\calS_n)} 
    \geq \frac{c_n}{\mu(\calS_n)}\enspace.
  \end{align*}

  So, the probability that $y$ is not dominated by any point in
  $S_n\cap\calS_n$ is $\left(1 - \frac{c_n}{\mu(\calS_n)}\right)^{z}$. Since this bound
  holds regardless of the specific value of $y$, applying Wald's
  identity,
  \begin{align*}
    \Pr[\E'_n]
    &\leq E[\text{number of samples in $R_n$ not dominated by $S_n\cap\calS_n$}]\\
    &= \left(1 - \frac{c_n}{\mu(\calS_n)}\right)^{z} \cdot E[|R_n|]\\
    &\leq \exp\left\{-\frac{z c_n}{\mu(\calS_n)}\right\} \cdot E[|R_n|]\enspace.
  \end{align*}
\end{proof}



  




\noindent
Finally we are ready to prove Theorem~\ref{thm:skyline-uniform}.

\begin{proofof}{Theorem~\ref{thm:skyline-uniform}}
  Using Lemma~\ref{lem:skyline-OPT}, we have that
  \begin{equation*}
    E_\mu[\Gap_n / \Gap^*_n] 
    \leq 1 + \frac{1}{c_n}(\Pr[\E_n  :  z \geq \alpha(n/4)] 
    + \Pr[z < \alpha(n/4)])\enspace.
  \end{equation*}
  where $z=|S_n\cap\calS_n|$. By Lemma~\ref{lem:const-frac}, the
  second term in the parentheses is $o(c_n)$. Event $\E_n$ implies
  event $\E'_n$. So applying Lemma~\ref{lem:dominate-rejects} and Fact
  \ref{fact:rejects-bound} we get
  \begin{align*}
    \Pr[\E_n  :  z \geq \alpha(n/4)] 
    &\leq \exp\left\{-\frac{z c_n}{\mu(\calS_{n})}\right\} \cdot E[|R_n|]\\
    &\leq \exp\{-\alpha(n/4)\}\cdot 
    \frac{n}{c_n}\\
    &= o(c_n)\enspace,
  \end{align*}
  since $c_i=\Omega(1/{\mbox {poly}}(i))$.
  \end{proofof}

\paragraph{General product distributions.}
We now consider the skyline model with $\mu$ being an arbitrary
product distribution. In particular, for a point $(a,b)\in X$, let
$\mu(a,b) = \mu^x(a)\mu^y(b)$ for one-dimensional measures $\mu^x$ and
$\mu^y$. Our proof of Theorem~\ref{thm:skyline-product} is nearly
identical to our argument for the uniform case. We note first that as
before
\[
 E_\mu\left[\frac{\Gap_n}{\Gap^*_n}\right] \leq 1
 + \frac{1}{c_n}\Pr[\E_n]\enspace. \]

To bound the probability of $\E_n$, we give a reduction from the
product measure setting to the uniform measure setting. In particular,
consider mapping $X$ into $X'=[0,1]\times[0,1]$ by mapping a point
$(a,b)\in X$ to $(\mu^x(a),\mu^y(b))\in X'$. Then it is easy to see
that $\calS_i$ in $X$ gets mapped to $\calS_i$ in $X'$ for the same
sequence $\{c_i\}$. Then, the probability of the event $\E_n$ under
the transformation remains the same as before, and is once again
$o(c_n)$.

\subsection{General spaces}
\label{sec:skyline-arbitrary}
In this section, we show how our results from the skyline model
generalize and prove Theorem \ref{thm:skyline-arbitrary}. 


Once again we note that Lemmas \ref{lem:skyline-OPT} and
\ref{lem:dominate-rejects} carry over to this setting. So our main
approach is to bound the probability of the event $\E'_n$. The main
difference from the previous analysis is that the best bound on the
size of $\calS_i$ we can obtain is $c_i\le\mu(\calS_i)\le 1$. This
means that we can no longer claim that a constant fraction of the
samples in $S_n$ belong to $\calS_n$. Instead we will show that under
a stronger condition on the sequence $\{c_i\}$, namely
$c_i=\Omega(1/i^\epsilon)$, the number of samples in $\calS_n$ is
$\Omega(nc_n)$ with a high probability. This will suffice to give us
the bound we need. In particular, we have the following weaker version of Lemma \ref{lem:const-frac}.




\begin{lemma}
\label{lem:weak-const-frac}
  For $c_i = \Omega(1/i^{\epsilon})$ with $\epsilon < 1$, we have that 
  \[\Pr[|S_n \cap \calS_n| < nc_n/2] = o(c_n)\enspace. \]
\end{lemma}

\begin{proof}
  For all $i \leq n$, we have that 
  \[\Pr_{x \sim \mu} [x \in \calS_{n}  :  x \in \calS_i]
  \geq \mu(\calS_n) \geq c_n\enspace.
  \]
  Thus, by Chernoff bounds,
  \[\Pr\left[ |S_n\cap\calS_n| < \frac{nc_n}{2}\right] 
  \leq \exp\{-\Omega(n^{1 - \epsilon})\} = o(c_n)\enspace.
  \]
\end{proof}

\noindent
Finally, we use the previous approach to prove Theorem \ref{thm:skyline-arbitrary}.

\begin{proofof}{Theorem \ref{thm:skyline-arbitrary}}
  By Lemma \ref{lem:skyline-OPT}, we have
  \[    E_\mu[\Gap_n / \Gap^*_n] 
    \leq 1 + \frac{1}{c_n}(\Pr[\E_n  :  z \geq \alpha(n/4)] 
    + \Pr[z < \alpha(n/4)])\enspace,
    \]
    where $z = |S_n \cap \calS_n|$.
    Using Fact~\ref{fact:rejects-bound} and
    Lemma~\ref{lem:dominate-rejects}, we have that
  \begin{align*}
    \Pr[\E_n  :  z \geq nc_n/2]& \le \Pr[\E'_n  : 
    z \geq nc_n/2] \\
    & \le \exp\left\{-\frac{nc_n^2}{2\mu(\calS_n)}\right\}\cdot E[|R_n|]\\
    & \leq \exp\{-n\cdot n^{-2(1/2-\Omega(1))}\} \cdot n/c_n\\
    & = o(c_n)\enspace,
\end{align*}
where the last inequality follows by
$c_i=i^{-(1/2-\Omega(1))}$. Together with Lemma
\ref{lem:weak-const-frac}, this proves Theorem
\ref{thm:skyline-arbitrary}.

\end{proofof}

\bibliographystyle{plain}
\bibliography{random}

\appendix
\section{Missing proofs}
\label{sec:appendix}

We present the technical proofs we skipped over in the main article.

\subsection{Sample selection in one dimension}
In section \ref{opt-exp-mean-gap:sec}, we had some technical
derivations that we prove in detail here.

\begin{proofof}{Lemma~\ref{fact:opt-mean-gap}}
Let $G(x)=x^k$ be the distribution of the gap of any sample. The
expectation of the $m$th order statistic for $t_n$ samples from this
distribution is as follows. See Claim~\ref{clm:1} below for a proof.
\[\E[x_{(m)}]= \frac{\Gamma(t_n+1)\Gamma(m+1/k)}{\Gamma(t_n +1 + 1/k)\Gamma(m)}.\]

So, we have that the expected mean gap of the $n$ smallest-gap
subsamples out of $t_n$ samples is
\begin{align*}
  \E\left[\frac{\sum_{i=1}^nx_{(i)}}{n}\right] 
  &= (1/n)\sum_{i=1}^n\E[x_{(i)}]\\
  &= (1/n) \frac{\Gamma(t_n+1)}{\Gamma(t_n +1
    + 1/k)}\sum_{m=1}^n \frac{\Gamma(m+1/k)}{\Gamma(m)}\\
  \intertext{we postpone the proof of the following step to Claim~\ref{clm:2} below,}
  &= (1/n)\frac{\Gamma(t_n+1)}{\Gamma(t_n +1
    + 1/k)}\frac{n\Gamma(n + 1 + 1/k)}{(1 + 1/k)\Gamma(n + 1)}\\
  &\sim \frac{1}{1 + 1/k} \left( \frac{n}{t_n + 1} \right)^{1/k},
\end{align*}
where the last line follows by equation 6.1.46 of
\cite{abramowitz1965handbook}. 
\end{proofof}

We now prove the two claims missing in the above proof.

\begin{claim}
\label{clm:1}
The expectation of the $m$th order statistic with
$F(x) = x^k$ is
\[\E[x_{(m)}]= \frac{\Gamma(t_n+1)\Gamma(m+1/k)}{\Gamma(t_n +1 + 1/k)\Gamma(m)}.\]
\end{claim}

\begin{proof}
From page 236 of \cite{tWIL62a}, we have that
the $m$th order statistic of a sample of size $t_n$ from a population
having continuous distribution function $F(x)$ and probability
distribution function $f(x)$ has the probability distribution
function:
\[\frac{\Gamma(t_n+1)}{\Gamma(m)\Gamma(t_n-m+1)}[F(x_{(m)})]^{m-1}[1 -
F(x_{(m)}]^{t_n-m}f(x_{(m)})dx_{(m)}\]

So, we have that the expectation of the $m$th order statistic with
$F(x) = x^k$ is
\begin{align*}
  \E[x_{(m)}]
  &= \int_0^1x_{(m)}\cdot\frac{\Gamma(t_n+1)}{\Gamma(m)\Gamma(t_n-m+1)}\cdot[F(x_{(m)})]^{m-1}\cdot[1
  - F(x_{(m)}]^{t_n-m}\cdot f(x_{(m)})dx_{(m)}\\
  &= \frac{\Gamma(t_n+1)}{\Gamma(m)\Gamma(t_n-m+1)}\int_0^1x_{(m)}[x_{(m)}^k]^{m-1}[1
  - x_{(m)}^k]^{t_n-m}kx_{(m)}^{k-1}dx_{(m)}\\
  &=  \frac{\Gamma(t_n+1)}{\Gamma(m)\Gamma(t_n-m+1)}
  \int_0^1x_{(m)}^{k(m-1) + 1}[1 -
  x_{(m)}^k]^{t_n-m}kx_{(m)}^{k-1}dx_{(m)}\\
  \intertext{using the substitution $y = x^k$, we have}
  &=  \frac{\Gamma(t_n+1)}{\Gamma(m)\Gamma(t_n-m+1)}
  \int_0^1y^{m-1+1/k}[1 - y]^{t_n - m} dy\\
  &= \frac{\Gamma(t_n+1)}{\Gamma(m)\Gamma(t_n-m+1)} \cdot B(m +
  1/k,t_n -m + 1)
  \intertext{where $B()$ is the Beta function}
  &= \frac{\Gamma(t_n+1)}{\Gamma(m)\Gamma(t_n-m+1)} \cdot 
  \frac{\Gamma(m+1/k)\Gamma(t_n - m + 1)}{\Gamma(t_n +1 + 1/k)}\\
  &= \frac{\Gamma(t_n+1)\Gamma(m+1/k)}{\Gamma(t_n +1 + 1/k)\Gamma(m)}\\
\end{align*}
\end{proof}

\begin{claim}
\label{clm:2}
  \[\sum_{m=1}^n\frac{\Gamma(m+1/k)}{\Gamma(m)} = \frac{n\Gamma(n + 1 + 1/k)}{(1/k + 1)\Gamma(n + 1)}. \]  
\end{claim}
\begin{proof}
To simplify a sum involving gamma functions, we can use the idea that
\[
  \Gamma(s) = \int_0^\infty t^{s-1} e^{-t} dt
\]
Now, we apply the Laplace transform. Using $s = m + \alpha$, and then interchanging summation and integration, we get
\begin{align*}
  \sum_{m=1}^n \Gamma(m+\alpha)/\Gamma(m)
  &= \int_0^\infty t^\alpha e^{-t} \sum_{m=0}^{n-1} \frac{t^m}{m!} dt\\
  &= \int_0^\infty t^\alpha \sum_{m=0}^{n-1} e^{-t} \frac{t^m}{m!}  dt\\
  &= \int_0^\infty t^\alpha \Pr[ \operatorname{Poisson}(t) < n ] dt\\
  &= \int_0^\infty t^\alpha \Pr[ G(n) \geq t ] dt,
\end{align*}
where $G(n)$ is a $\operatorname{Gamma}(n)$ random variable.
The last identity follows by the Poisson process.
Plugging in
\[
\Pr[G(n) \geq t] = \int_t^\infty u^{n-1} e^{-u} du
\]
and then interchanging the order of integration gives an integral
that evaluates to 
\[ \frac{n\Gamma(n + 1 + 1/k)}{(1/k + 1)\Gamma(n + 1)}.\]
\end{proof}

\end{document}